\documentclass{article}
\usepackage{latexsym}
\usepackage{amssymb}
\usepackage{amscd}
\usepackage{bm}
\usepackage{amsmath,amsthm}
\usepackage{graphicx}
\usepackage{amsmath}
\usepackage{authblk}
\usepackage{amsthm,amsmath,amssymb}
\usepackage{mathrsfs}
\usepackage{color}
\usepackage[sort,numbers]{natbib}
\usepackage{hyperref}
\hypersetup{
    colorlinks=true, 
    linkcolor=blue, 
    anchorcolor=blue, 
    citecolor=blue 
}

\everymath{\displaystyle}
\begin{document}
\renewcommand{\theequation}{\arabic{section}-\arabic{equation}}
\newtheorem{thm}{Theorem}[section]
\newtheorem{thmsec}{Theorem}[section]
\newtheorem {rmk}[thm]{Remark}
\newtheorem {rmksec}[thmsec]{Remark}
\newtheorem{lem}[thm]{Lemma}
\newtheorem{prop}{Proposition}
\newtheorem{coro}[thm]{Corollary}
\hoffset = -2.5truecm \voffset = -2.0truecm
\def\cdot{{\scriptstyle\,\bullet\,}}
\renewcommand{\theequation}{\arabic{section}-\arabic{equation}}

\def\red{\color{red}}
\def\blue{\color{blue}}

\title{\bf Nonisospectral deformations of noncommutative Laurent biorthogonal polynomials and matrix discrete Painlev\'{e}-type equations}
\author[1]{Dan Dai}
\affil[1]{Department of Mathematics, City University of Hong Kong Tat Chee Avenue, Kowloon, Hong Kong}
\author[1]{Xiaolu Yue \thanks{xiaolyue@cityu.edu.hk}}
\date{}
\maketitle{}
\begin{abstract}
In this paper, we establish a connection between noncommutative Laurent biorthogonal polynomials (bi-OPs) and matrix discrete Painlev\'{e} (dP) equations. We first apply nonisospectral deformations to noncommutative Laurent bi-OPs
to obtain the noncommutative nonisospectral mixed relativistic Toda lattice and its Lax pair. Then, we perform a stationary reduction on this Lax pair to obtain a matrix dP-type equation.
The validity of this reduction is demonstrated through a specific choice of weight function and the application of quasideterminant properties. In the scalar case, our matrix dP equation reduces to the known alternate dP II equation.\\
\textbf{Keywords:} noncommutative Laurent biorthogonal polynomials; nonisospectral deformations; matrix discrete Painlev\'{e} equations; stationary reduction; quasideterminant
\end{abstract}

\section{Introduction}

Orthogonal polynomials (OPs) constitute a classical subject with a rich history spanning two centuries; for example, see monographs by Szeg\H{o} \cite{MR372517}, Chihara \cite{MR481884} and Ismail \cite{MR2542683}. Classical families, such as Hermite polynomials, Laguerre polynomials, and Jacobi polynomials, have wide application across mathematics and physics, including approximation theory, random matrix theory, integrable systems, quantum mechanics, and so on. Although the classical orthogonal theory can be traced back to the 19th century, ongoing discoveries of new applications continue to drive the field, leading to various significant generalizations. One of the famous examples is the family of biorthogonal polynomials (bi-OPs), which include Cauchy bi-OPs, partial-skew OPs, and Laurent bi-OPs,  among others. For instance, inspired by the study of multi-peakon solutions of the Degasperis-Procesi shallow water wave equation, Bertola et al introduced Cauchy bi-OPs in \cite{bertola2010cauchy}. 
Chang et al. proposed the concept of partial-skew OPs \cite{chang2018partial}, which are associated with the Bures random ensemble. As an extension of OPs on the unit circle \cite{kharchev1997faces}, Laurent bi-OPs are  introduced in the study of the two-point Pad\'{e} approximation problem \cite{jones2006survey}.
 It also appears in the setting of the pentagram map defined on polygons in the projective space, with the full discrete  relativistic Toda (rToda) lattice being classified as the leapfrog map.
Recently, \cite{wang2022generalization} presented generalized Laurent bi-OPs  and achieved both the generalized positive and negative rToda. This work was partially inspired by the investigation of the so-called coupled pentagram map in \cite{MR4832540}.


It is well-known that there is a profound connection between orthogonal polynomials with semi-classical weight functions and both continuous and discrete Painlev\'e equations; see Van Assche \cite{van2017orthogonal} and references therein. The relationship between biorthogonal polynomials and integrable systems has also been extensively studied in recent years. For example, by introducing a time evolution into the weight function of the Cauchy bi-OPs, the corresponding recurrence coefficients satisfy a CKP-type Toda equation \cite{yueXL2023thesis}. Similarly, the BKP-type Toda equation can be generated from partial-skew OPs \cite{chang2018partial,yueXL2023thesis}. Furthermore, Laurent bi-OPs can also be used to derive semi-discrete and fully discrete rToda lattice \cite{kharchev1997faces}. In particular, the semi-discrete rToda lattice is associated with both the Lotka-Volterra lattice and $R_I$ lattice \cite{vinet1998integrable}. For a given power series, the fully discrete rToda lattice can be used to design a new Pad\'{e} approximation algorithm \cite{minesaki2001discrete}.
In this paper, we contribute to this field of research by investigating the connection between noncommutative Laurent bi-OPs and discrete Painlev\'{e} (dP) equations. The dP equations are nonlinear, non-autonomous, second-order ordinary difference equations that pass an integrability criterion known as singularity confinement \cite{grammaticos1991integrable}. Since they reduce to the continuous versions in a suitable continuum limit, dP equations are regarded as discrete counterparts to the classical Painlev\'e differential equations; see, for example, Joshi \cite{MR3931704}. Furthermore, dP equations possess significant applications in numerous areas, such as geometry, reductions of lattice equations, quantum gravity, and certain discrete gap probabilities in random partitions; for example, see \cite{MR1979052,fokas1991discrete,hone2002lattice, grammaticos1999discrete}.

Various methods for deriving dP equations have been proposed in the literature, including two effective techniques: the compatibility method based on orthogonality \cite{van2017orthogonal} and the stationary reduction method for nonisospectral flow \cite{levi1992non}. 
The first method generally requires choosing an appropriate semi-classical weight function to derive the corresponding structure relation, and then dP equations are deduced by applying the compatibility condition between the recurrence relation and the structure relation. However, for bi-OPs, since one can not simply add a semi-classical factor in the weight function, this approach seems to be inapplicable. The second method involves a direct stationary reduction of nonisospectral equations to derive the dP equations. However, a fundamental limitation of this methodology is that it does not provide the solutions to the dP equations, nor does it justify the stationary reduction from the perspective of solutions. It is worth noting that Yue, Chang, and Hu recently refined this method in their application of nonisospectral deformations to OPs in \cite{yueXL2023thesis}. In their work, they first introduce a nonisospectral deformation without specifying a particular weight function. Then, they use the compatibility condition between the recurrence relation and the time evolution to derive nonisospectral integrable equations and their corresponding Lax pairs. A subsequent stationary reduction of these systems yields the dP equations.
Moreover, they go beyond this by providing  quasideterminant solutions to the dP equations and rigorously justifying the stationary reduction through the explicit construction of a specific weight function.

As a noncommutative generalization of OPs, matrix OPs (MOPs) were introduced by Krein \cite{MR34964}. 
Research interest in MOPs has grown significantly in recent years. One of the key developments is the formulation of a Riemann-Hilbert (RH) problem for MOPs in \cite{cassatella2012riemann}, generalizing the seminal results of Fokas, Its, and Kitaev \cite{fokas1992isomonodromy} to the matrix context. Subsequently, Cafasso established double integral representations for the Christoffel-Darboux kernels related to two Hermite-type MOPs and demonstrated that their associated Fredholm determinants are connected to a specific RH problem in \cite{cafasso2014non}. 
Furthermore, analogous to the scalar case, MOPs on the real line have been shown to satisfy noncommutative versions of integrable hierarchies such as the Toda and Volterra lattices \cite{li2024matrix}. Similarly, MOPs on the unit circle lead to the Ablowitz-Ladik hierarchy \cite{cafasso2009matrix}. These connections naturally extend to the study of noncommutative dP equations; for example, see \cite{cassatella2012riemann,adler2020painleve}.


The connection between Laurent bi-OPs and the leapfrog map can be generalized to the noncommutative setting
\cite{wang2025noncommutative}, in which the leapfrog map is recognized as a one-dimensional counterpart of the pentagram map introduced in \cite{gekhtman2016integrable}. 
Moreover, the noncommutative rToda has been successfully characterized through the application of noncommutative Laurent bi-OPs and noncommutative leapfrog map in prior studies \cite{wang2025noncommutative}.
Additionally, dP equations have been studied in the context of scalar Laurent bi-OPs \cite{yue2022laurent} and were derived from their generalized counterparts via stationary reduction method based on nonisospectral deformation \cite{yueXL2023thesis}. However, to the best of our knowledge, the relationship between noncommutative Laurent biorthogonal polynomials and matrix dP equations remains unexplored. This paper aims to address this gap by employing the refined nonisospectral deformation approach introduced in \cite{yueXL2023thesis}.

The structure of this paper is as follows.  In Section \ref{preliminears}, we provide a brief introduction to the basic properties of quasideterminants and noncommutative Laurent bi-OPs. In Section \ref{section-toda}, we perform nonisospectral deformations on noncommutative Laurent bi-OPs, leading to the derivation of the noncommutative nonisospectral mixed rToda lattice. Moreover, we apply stationary reduction to obtain the matrix dP. We then select a specific weight function and use the properties of quasideterminants to justify the reduction. 
Finally, we provide a conclusion and a discussion of this work in Section \ref{conclusion-section}.

\section{Preliminaries}\label{preliminears}
In this section, we will give a brief introduction of quasideterminants, noncommutative Laurent bi-OPs, particularly matrix Laurent bi-OPs and their related properties.

\subsection{Quasideterminants}


Consider an $N\times N$ matrix $A=(a_{i,j})_{i,j=1}^{N}$ with entries defined over a noncommutative ring. Let $A^{i,j}$ be the submatrix of $A$ obtained by removing the $i$th row and the $j$th column from $A$. When all the inverses $\left(A^{i, j}\right)^{-1}$ exist for $i,j=1,\cdots,N$, the matrix $A$ admits
$N^2$ well-defined quasideterminants, denoted by $|A|_{i,j}$.
They are defined recursively by
\begin{equation}
|A|_{i,j}=a_{i,j}-r_i^j \left(A^{i,j}\right)^{-1}c_j^i=\begin{vmatrix}
    A^{i,j} & c_j^i \\ r_i^j & \boxed{a_{i,j}}\end{vmatrix},  \quad A^{-1}= \left(|A|_{j,i} ^{-1}\right),
\end{equation}
where $r_i^j$ represents the $i$th row of $A$ with the $j$th element removed, $c_j^i$ represents the $j$th column of $A$ with the $i$th element removed. In fact, if the entries $a_{i,j}$ in $A$ commute, then we have 
\begin{equation*}
\begin{vmatrix} A \end{vmatrix} _{i,j}  = (-1)^{i+j} \frac{\mbox{det}(A)}{\mbox{det}(A^{i,j})} . 
\end{equation*}

Let $A, B, C$ and $d$ be functions of the independent variable $t$. Then, we have
\begin{equation*} 
    \begin{vmatrix} A & B \\ C & \boxed{d} \end{vmatrix}'=d'-C' A^{-1}B-C A^{-1}B'+C A^{-1}A'A^{-1}B,
\end{equation*}
where $'$ denotes the derivative with respect to $t$. If we incorporate the identity matrix expressed as $\sum_{k=0}^{N-1}e_k^Te_k$, where $e_k$ is the row vector of length $N$ with a value of $1$ at the  $(k+1)$-th position and $0$ elsewhere. Let $A_k$ represent the $k$-th column of the matrix $A$, then we get
\begin{align}\label{derivative-formulas-for-quasideterminants}
    \left|
\begin{array}{cc}
A&B\\
C&\boxed{d}
\end{array}\right|'=&\left|
\begin{array}{cc}
A&B'\\
C&\boxed{d'}
\end{array}\right|+\sum_{k=0}^{N-1}\left(\left|
\begin{array}{cc}
A&(A_k)'\\
C&\boxed{(C_k)'}
\end{array}\right|
\left|\begin{array}{cc}
A&B\\
    e_k&\boxed{0}
\end{array}\right|\right)\\
=&\left|
\begin{array}{cc}
A&B\\
C'&\boxed{d'}
\end{array}\right|+\sum_{k=0}^{N-1}\left(\left|
\begin{array}{cc}
A&e_k^{T}\\
C&\boxed{0}
\end{array}\right|
\left|\begin{array}{cc}
A&B\\
    (A^k)'&\boxed{(B^k)'}
\end{array}\right|\right).
\end{align}
In addition, there are numerous other important properties of quasideterminants; for details, please refer to \cite{gelfand2005quasideterminants,gilson2007direct,MR1328259,li2008quasideterminant}.

\subsection{Noncommutative Laurent bi-OPs}  \label{section-matrix-orthogonal-polynomials}

Gelfand et al. introduced the theory of noncommutative OPs in \cite{gelfand1994noncommutative}, establishing a formal analogy with MOPs. The theory of noncommutative OPs has been extended to noncommutative bi-OPs, such as noncommutative Laurent bi-OPs and Cauchy bi-OPs \cite{ariznabarreta2022multivariate,wang2025noncommutative, li2025matrix}. Here, we primarily study noncommutative Laurent bi-OPs.

Let $R$ be a skew field generated by the unity $1$ and the formal moments $\{m_i\}_{i=-\infty}^{\infty}$. Correspondingly, the formal power series (resp., polynomials) in $\lambda$ with coefficients from the skew field $R$ are denoted by $R[[\lambda]]$ (resp.,$R[\lambda]$). Furthermore, this skew field is equipped with an involution $R\rightarrow R^*$, which satisfies $(a_i)^*=a_i^*$ for the coefficients. 
This involution can be generalized to the polynomial ring in a manner consistent with $R[\lambda]\rightarrow R^*[\lambda^{-1}]$, so that
\begin{align*}
   \big( \sum_i a_i\lambda^i\big)^*=\sum_i a_i^*\lambda^{-i}.
\end{align*}
Thus, we define an inner product $\langle \cdot ,\cdot \rangle$: $R[[\lambda]]\times R[[\lambda]]\rightarrow{R}$ as
\begin{align}\label{inner defi}
   \big \langle
    \sum_i a_i\lambda^i,\sum_j b_j\lambda^j\big\rangle=\sum_{i,j} a_i m_{i-j}b_j^*.
\end{align}
The following properties for the inner product can be easily verified.
\begin{enumerate}
    \item $\langle \alpha_1 p_1(\lambda)+\alpha_2 p_2(\lambda), q(\lambda)\rangle=\alpha_1\langle p_1(\lambda),q(\lambda)\rangle +\alpha_2\langle  p_2(\lambda), q(\lambda)\rangle$;
    \item $\langle p(\lambda),\beta_1 q_1(\lambda)+\beta_2 q_2(\lambda)\rangle=\langle p(\lambda),q_1(\lambda)\rangle \beta_1^*+\langle p(\lambda),q_2(\lambda)\rangle \beta_2^*$;
    \item $ \langle \lambda p(\lambda),q(\lambda)\rangle =\langle p(\lambda),\lambda^{-1}q(\lambda)\rangle $,
\end{enumerate}
where $p_1(\lambda),\,p_2(\lambda),\,q_1(\lambda),\,q_2(\lambda),\, p(\lambda),\, q(\lambda) \in R[[\lambda]]$ and $\alpha_1,\, \alpha_2,\,\beta_1,\,\beta_2 \in R$.

Monic noncommutative Laurent bi-OPs denote two families of
polynomials $\{P_n(\lambda)\}_{n\in N}$
and $\{Q_n(\lambda)\}_{n\in N}$ satisfying the biorthogonality condition
\begin{align}\label{orthogonal-pi}
    \langle P_n(\lambda),Q_m(\lambda) \rangle=H_n\delta_{nm},
\end{align}
where $H_n\in R$ is a normalization factor and the leading coefficients of
$P_n(\lambda)$ and $Q_n(\lambda)$ are the units in $R$.

By using the biorthogonality condition \eqref{orthogonal-pi}, it is not hard to see that
$\{P_n(\lambda)\}_{n\in N}$
and $\{Q_n(\lambda)\}_{n\in N}$ can be given in terms of quasideterminant expressions

\begin{subequations}
\begin{align}
        P_n(\lambda)=&\left|\begin{matrix}
         m_{0}    & \cdots &  m_{1-n}    &  1\\
          \vdots& \ddots           &  \vdots   &  \vdots\\
        m_{n-1}  & \cdots &  m_{0}    &  \lambda^{n-1}\\
        m_n      & \cdots &  m_{1}  &  \boxed{\lambda^n}
    \end{matrix}
        \right|:=\left|
\begin{array}{cc}
\Lambda_{n-1}&L^T\\
\theta_n^n&\boxed{\lambda^n}
\end{array}\right|,\label{the-experssion-of-P_n}\\
  (Q_n(\lambda))^*=&\left|\begin{matrix}
         m_{0}    & \cdots &  m_{1-n}    &  m_{-n}\\
          \vdots& \ddots          &  \vdots   &  \vdots\\
        m_{n-1}  & \cdots &  m_{0}    &  m_{-1}\\
        1      & \cdots &  \lambda^{1-n}  &  \boxed{\lambda^{-n}}
    \end{matrix}
        \right|:=\left|
\begin{array}{cc}
\Lambda_{n-1}&(\tilde{\theta}_{-n}^{n-1})^T\\
\tilde{L}&\boxed{\lambda^{-n}}
\end{array}\right|,
    \end{align}
with
    \begin{align*}
 &\Lambda_{n-1} = \left(\begin{matrix}
        m_{0}    & \cdots &  m_{1-n}   \\
          \vdots   &\ddots        &  \vdots   \\
        m_{n-1}  & \cdots &  m_{0}    
    \end{matrix}\right), \\
&\theta_j^i = (m_j,m_{j-1},\cdots,m_{j-i+1}), \quad \tilde{\theta}_j^i = (m_j,m_{j+1},\cdots,m_{j+i}), \\
& L = (1,\lambda,\cdots,\lambda^{n-1}), \quad \tilde{L} = (1,\lambda^{-1},\cdots,\lambda^{1-n}).  
\end{align*}
And $H_n$ has the quasideterminant representation
\begin{align}\label{hn}
H_n=\left|\begin{matrix}
         m_{0}    & \cdots &  m_{1-n}    &  m_{-n}\\
          \vdots   &\ddots        &  \vdots   &  \vdots\\
        m_{n-1}  & \cdots &  m_{0}    &  m_{-1}\\
        m_{n}      & \cdots &  m_{1}  &  \boxed{m_{0}}
    \end{matrix}
        \right|:=\left|
\begin{array}{cc}
\Lambda_{n-1}&(\tilde{\theta}_{-n}^{n-1})^T\\
\theta_n^n&\boxed{m_0}
\end{array}\right|.
\end{align}
\end{subequations}
In addition, we can utilize biorthogonality \eqref{orthogonal-pi} to derive the following proposition.

\begin{prop}\label{tree terms}
Noncommutative Laurent bi-OPs $\{P_n(\lambda)\}_{n\in N}$ satisfy the three-term recurrence relation
\begin{equation}\label{3-term}
\lambda (P_n(\lambda)+a_nP_{n-1}(\lambda))=P_{n+1}(\lambda)+b_n P_n(\lambda),    
\end{equation}
where the recurrence coefficients $a_n$ and $b_n$ have quasideterminant expressions
\begin{align}\label{an bn}
        a_n=-\tau_n \tau_{n-1}^{-1},\quad  b_n=a_n H_{n-1}H_n^{-1},
    \end{align}
where
\begin{align}\label{quasi-tau}
    \tau_n=\left|\begin{matrix}
         m_{0}    & \cdots &  m_{1-n}    &  m_{1}\\
          \vdots   & \ddots       &  \vdots   &  \vdots\\
        m_{n-1}  & \cdots &  m_{0}    &  m_{n}\\
        m_n      & \cdots &  m_{1}  &  \boxed{m_{n+1}}
    \end{matrix}
        \right|:=\left|
\begin{array}{cc}
\Lambda_{n-1}&(\tilde{\theta}_{1}^{n-1})^T\\
\theta_n^n&\boxed{m_{n+1}}
\end{array}\right|.
\end{align}

\end{prop}

\begin{proof}
    Setting
    \begin{equation*}
        \langle  P_n(\lambda)+a_n P_{n-1}(\lambda), \lambda^{-1}\rangle=0,
    \end{equation*}
    we have
    \begin{equation*}
        a_n=-\langle P_n,\lambda^{-1}\rangle \langle P_{n-1},\lambda^{-1}\rangle ^{-1}.
    \end{equation*}
 Furthermore, from the quasideterminant representation of $P_n(\lambda)$, it is straightforward to derive $a_n=-\tau_n \tau_{n-1}^{-1}$. The polynomial $\lambda(P_n(\lambda)+a_n P_{n-1}(\lambda))$ is monic of degree $n+1$ with respect to $\lambda$, so it can be expanded in terms of noncommutative Laurent bi-OPs $P_k$ with $0\leq k \leq n+1$, 
 \begin{equation*}
     \lambda(P_n(\lambda)+a_n P_{n-1}(\lambda))=P_{n+1}(\lambda)+\sum_{k=0}^n \alpha_{n,k}P_k(\lambda).
 \end{equation*}
Taking the inner product of both sides of the above equation with $Q_k(\lambda)$ and using the orthogonality yields:
\begin{align}
    \alpha_{n,k}=\langle P_{n}(\lambda)+a_nP_{n-1}(\lambda),\lambda^{-1}Q_k(\lambda)\rangle \langle P_k(\lambda),Q_k(\lambda) \rangle^{-1}
   =0,\nonumber
\end{align}
and  $\alpha_{n,k}=0$ when $k < n$. So we have
\begin{equation*}
    \lambda (P_n(\lambda)+a_nP_{n-1}(\lambda))=P_{n+1}(\lambda)+b_n P_n(\lambda),
\end{equation*}
and 
\begin{equation*}
    b_n=\alpha_{n,n}=\langle P_{n}(\lambda)+a_nP_{n-1}(\lambda),\lambda^{-1}Q_n(\lambda)\rangle \langle P_n(\lambda),Q_n(\lambda) \rangle^{-1}=a_nH_{n-1}H_n^{-1}.
\end{equation*}
Therefore, we have completed the proof.
\end{proof}

 We rewrite \eqref{3-term} in matrix form
\begin{align}\label{lax-r}
    \lambda A P=BP,
\end{align}
where $P=(P_0(\lambda),P_1(\lambda), P_2(\lambda),\ldots)^T$, and 
\begin{align}
A=\left(
\begin{array}{cccccc}
 1 &  &  &  &   \\
 a_1 & 1 &  &  &   \\
  & a_2 &1 &  &   \\
  &  &  \ddots & \ddots & \ \\
  &  &  &  \ddots & \ddots
\end{array}
\right),\quad
B=\left(\begin{array}
{cccccc}
 b_0 & 1 &  &  &   \\
  & b_1 & 1 &  &   \\
  &  & b_2 & 1 &   \\
  &  &   & \ddots & \ddots \\
  &  &   &  & \ddots \\
\end{array}
\right). \label{def:A-B}
\end{align}




\subsection{Matrix Laurent bi-OPs} \label{sec:Matrix Laurent bi-OPs}

As a special case of noncommutative Laurent bi-OPs, the study of matrix Laurent bi-OPs $\{P_n(\lambda)\}_{n\in \mathbb{N}}$ arises naturally. In this setting, the underlying skew field $R$ can be chosen as a subalgebra of $\mathbb{R}_+^{p\times p}$, where the unity $1$ corresponds to the identity matrix $\mathbb{I}_p$, and the formal moments $\{m_i\}_{i=-\infty}^{\infty}\in R$ are endowed with a natural involution $R\rightarrow R^T$.

A pair of monic matrix polynomials $\{P_n(\lambda)\}_{n\in \mathbb{N}}$  and $\{Q_n(\lambda)\}_{n\in \mathbb{N}}$ are called matrix Laurent bi-OPs 
if they satisfy the orthogonality condition 
\begin{align}\label{orthogonality condition}
    \langle P_n(\lambda),Q_m(\lambda) \rangle=\int P_{n}(\lambda)d\mu(\lambda) Q_{m}^{T}(\frac{1}{\lambda})=H_n\delta_{nm},
\end{align}
where $\mu(\lambda)$ is a semi-positive definite $p\times p$ matrix valued measure that ensures all the moments
\begin{align} \label{moments-matrix}
    m_i=\int \lambda^id\mu(\lambda)
\end{align}
exist. 

Naturally, 
$\{P_n(\lambda)\}_{n\in \mathbb{N}}$, $\{Q_n(\lambda)\}_{n\in \mathbb{N}}$ and $H_n$ can still be expressed in terms of a quasideterminant as that in \eqref{the-experssion-of-P_n}-\eqref{hn}, respectively. Additionally, using the biorthogonality condition \eqref{orthogonality condition}, it is easy to show that matrix Laurent bi-OPs $\{P_n(\lambda)\}_{n\in \mathbb{N}}$ satisfy the three-term recurrence relation \eqref{3-term}. The recurrence relation satisfied by 
$\{Q_n(\lambda)\}_{n\in \mathbb{N}}$ can also be derived using the biorthogonality (cf. \cite{wang2025noncommutative}). 

\section {Noncommutative nonisospectral mixed rToda  lattice and matrix discrete Painlev\'{e} equation}\label{section-toda}
\setcounter{equation}{0}
In this section, we investigate nonisospectral deformations of noncommutative Laurent bi-OPs. 
First, by employing the biorthogonality condition \eqref{orthogonal-pi}, we establish the time evolution relations for noncommutative Laurent bi-OPs.
Subsequently, through the compatibility condition between three-term recurrence relation \eqref{3-term} and the time evolution we obtained, the noncommutative nonisospectral mixed rToda lattice is derived. Finally, by applying stationary reduction to both the lattice and its Lax pair, we obtain the matrix discrete Painlev\'{e} equation along with its corresponding Lax pair.

\subsection{Noncommutative nonisospectral mixed rToda lattice} \label{subsection-toda-1}
Nonisospectral deformations of noncommutative Laurent bi-OPs are performed by introducing the measure $\mu(\lambda,t)$, in which the spectral parameter $\lambda(t)$ is also time-dependent. Specifically, we consider the spectral parameter $\lambda(t)$ satisfying 
\begin{equation} \label{derivative-of-x}
\frac{d \lambda(t)}{d t} =\alpha_0 \lambda(t),
\end{equation}
where $\alpha_0 \neq 0$.
In the noncommutative setting, the moment $m_i$ is defined via the inner product in \eqref{inner defi}. More precisely, we have
\begin{equation}\label{non mom}
    m_i(t) = \big \langle
    \lambda(t)^k, \lambda(t)^j\big\rangle, \qquad \textrm{with } k-j=i.
\end{equation}
In the following text, we will denote $\lambda(t)$ as $\lambda$ for brevity. We further assume that the moments satisfy the following time evolution 
\begin{align}\label{moments evolution}
\frac{d}{dt}m_j(t)=\alpha_0 jm_j(t)+\alpha_1 m_{j+1}(t)+\alpha_2 m_{j-1}(t),
\end{align}
where $\alpha_1,\alpha_2$ are two arbitrary constants.

Given assumptions \eqref{derivative-of-x} and \eqref{moments evolution}, we can present the time evolution of the noncommutative Laurent bi-OPs in the following lemma.

\begin{lem}
  With the assumptions \eqref{derivative-of-x} and \eqref{moments evolution}, the noncommutative Laurent bi-OPs  $\{P_n(\lambda,t)\}_{n\in \mathbb{N}}$  defined in \eqref{orthogonal-pi} depend on $t$ and satisfy the  time evolution 
  \begin{align}\label{p-times}
   & \frac{d}{dt}P_n(\lambda,t)+a_n\frac{d}{dt}P_{n-1}(\lambda,t)\nonumber\\
=&n\alpha_0 P_n(\lambda,t)+((n-1)\alpha_0a_n+\alpha_1a_n(b_{n-1} -a_{n-1})-\alpha_2b_n^{-1}a_n)P_{n-1}(\lambda,t)-\alpha_2 a_nb_{n-1}^{-1}a_{n-1}P_{n-2}(\lambda,t).
\end{align}
\end{lem}

\begin{proof}
    Let 
    $$P_n(\lambda,t)=\lambda^n+\gamma_{n,n-1}\lambda^{n-1}+\cdots+\gamma_{n,1}\lambda+\gamma_{n,0}.$$
    From the biorthogonality condition \eqref{orthogonal-pi}, we know that when $i=0,1,\cdots,n-1,$
    \begin{equation*}
         0=\langle P_n(\lambda,t),\lambda^i \rangle=m_{n-i}+\gamma_{n,n-1}m_{n-1-i}+\cdots+\gamma_{n,1}m_{1-i}+\gamma_{n,0}m_{-i}.
    \end{equation*}
    
Under assumption \eqref{moments evolution}, differentiating both sides of the above equation with respect to $t$ yields
\begin{align}
0=&\dot{\gamma}_{n,n-1}m_{n-1-i}+\cdots+\dot{\gamma}_{n,1}m_{1-i}+\dot{\gamma}_{n,0}m_{-i}+\alpha_0(nm_{n-i}+(n-1)\gamma_{n,n-1}m_{n-1-i}+\cdots+\gamma_{n,1}m_{1-i})\nonumber\\
    &-i\alpha_0(m_{n-i}+\gamma_{n,n-1}m_{n-1-i}+\cdots+\gamma_{n,0}m_{-i})+\alpha_1(m_{n-i+1}+\gamma_{n,n-1}m_{n-i}+\cdots+\gamma_{n,0}m_{1-i})\nonumber\\
    &+\alpha_2(m_{n-i-1}+\gamma_{n,n-1}m_{n-2-i}+\cdots+\gamma_{n,0}m_{-1-i})\nonumber\\
    =&\langle \frac{d}{dt}P_n(\lambda,t)-i\alpha_0P_n(\lambda,t)+\alpha_1\lambda P_n(\lambda,t)+\alpha_2\lambda^{-1}P_n(\lambda,t),\lambda^i \rangle\nonumber\\
    =&\langle \frac{d}{dt}P_n(\lambda,t)+\alpha_1\lambda P_n(\lambda,t)+\alpha_2\lambda^{-1}P_n(\lambda,t),\lambda^i \rangle,\label{br0}
\end{align}
where $\dot{}$ denotes the derivative with respect to $t$. With the aid of the biorthogonality condition \eqref{orthogonal-pi} and the recurrence relation \eqref{3-term},
we have
\begin{align}\label{br1}
\langle \alpha_1\lambda P_n(\lambda,t),\lambda^i \rangle =& \langle -\alpha_1\lambda a_n P_{n-1}(\lambda,t)+P_{n+1}(\lambda,t)+b_nP_{n}(\lambda,t),\lambda^i \rangle\nonumber\\
  =&- \langle \alpha_1\lambda a_n P_{n-1}(\lambda,t),\lambda^i\rangle\nonumber\\
  =&- \langle \alpha_1a_n(\lambda  P_{n-1}(\lambda,t)-P_n(\lambda,t)),\lambda^i\rangle,
\end{align}
and
\begin{align}\label{br2}
    \langle \alpha_2\lambda^{-1}P_n(\lambda,t),\lambda^i \rangle=&\alpha_2b_n^{-1}\langle \lambda^{-1}(P_{n+1}(\lambda,t)+b_nP_n(\lambda,t)),\lambda^i \rangle\nonumber\\
    =&\alpha_2b_n^{-1}\langle P_{n}(\lambda,t)+a_nP_{n-1}(\lambda,t),\lambda^i \rangle\nonumber\\
    =&\alpha_2b_n^{-1}a_n\langle P_{n-1}(\lambda,t),\lambda^i \rangle.
\end{align}
    Since $\frac{d\lambda}{dt}=\alpha_0\lambda$, it is straightforward to see that $\frac{d}{dt}P_n(\lambda,t)=n\alpha_0P_n(\lambda,t)+\sum_{j=0}^{n-1}\beta_jP_j(\lambda,t)$. Substituting \eqref{br1} and \eqref{br2} into \eqref{br0} gives
    \begin{equation*}
        \langle \sum_{j=0}^{n-1}\beta_jP_j(\lambda,t)-\alpha_1a_n(\lambda  P_{n-1}(\lambda,t)-P_n(\lambda,t))+\alpha_2b_n^{-1}a_n P_{n-1}(\lambda,t),\lambda^i \rangle =0.
    \end{equation*}
  Given the validity of the above equation for $i=0,1,\cdots,n-1$, and by applying the biorthogonality condition \eqref{orthogonal-pi}, we find that
  \begin{equation*}
    \sum_{j=0}^{n-1}\beta_jP_j(\lambda,t)=\alpha_1a_n(\lambda  P_{n-1}(\lambda,t)-P_n(\lambda,t))-\alpha_2b_n^{-1}a_n P_{n-1}(\lambda,t) , 
  \end{equation*}
  so
  \begin{equation}\label{time evo 1}
     \frac{d}{dt}P_n(\lambda,t)=n\alpha_0P_n(\lambda,t)+ \alpha_1a_n(\lambda  P_{n-1}(\lambda,t)-P_n(\lambda,t))-\alpha_2b_n^{-1}a_n P_{n-1}(\lambda,t).
  \end{equation}
This leads us directly to
\begin{align}
   & \frac{d}{dt}P_n(\lambda,t)+a_n\frac{d}{dt}P_{n-1}(\lambda,t)\nonumber\\
=&n\alpha_0 P_n(\lambda,t)+((n-1)\alpha_0a_n+\alpha_1a_n(b_{n-1}-a_{n-1})-\alpha_2b_n^{-1}a_n)P_{n-1}(\lambda,t)-\alpha_2 a_nb_{n-1}^{-1}a_{n-1}P_{n-2}(\lambda,t).\nonumber
\end{align}
Therefore, we have completed the proof.
\end{proof}

We can also re-express the above lemma in the concrete matrix case introduced in Section~\ref{sec:Matrix Laurent bi-OPs}. Under the time evolution condition \eqref{derivative-of-x} for the spectral parameter $\lambda$, we may further assume that the measure $d\mu(\lambda)$ in \eqref{orthogonality condition} can be written as $d\mu(\lambda) = w(\lambda,t)\,d\lambda$, where $w(\lambda,t)$ is a $p \times p$ matrix-valued weight function. Then the moments \eqref{moments-matrix} also depend on the time variable $t$ and can be rewritten as
\begin{align}
m_j(t) = \int \lambda^j \, d\mu(\lambda,t) = \int \lambda^j w(\lambda,t) \, d\lambda.
\end{align}
Then, we have
\begin{align}
\frac{d}{dt} m_j(t) &= \int \lambda^j \left( \alpha_0 j \, w(\lambda,t) + \frac{d}{dt} w(\lambda,t) + \alpha_0 w(\lambda,t) \right) d\lambda \nonumber \\
&= \alpha_0 j \, m_j(t) + \int \lambda^j \left( \frac{d}{dt} w(\lambda,t) + \alpha_0 w(\lambda,t) \right) d\lambda.
\end{align}
Therefore, the assumption \eqref{moments evolution} becomes
\begin{equation}
    \int \lambda^j \left(\frac{d}{dt} w(\lambda,t) + \alpha_0 w(\lambda,t) \right) d \lambda = \alpha_1 m_{j+t}(t) + \alpha_2 m_{j-1}(t) = \int \lambda^j (\alpha_1 \lambda + \alpha_2 \lambda^{-1}) w(\lambda,t) d \lambda.
\end{equation}
Clearly, a sufficient condition for the above identity is
\begin{align}\label{weight ass}
\frac{d}{dt} w(\lambda,t) + \alpha_0 w(\lambda,t) = (\alpha_1 \lambda + \alpha_2 \lambda^{-1}) w(\lambda,t).
\end{align}
Thus, we obtain the following corollary.

\begin{coro}
With the assumptions \eqref{derivative-of-x} and \eqref{weight ass}, the matrix Laurent bi-OPs  $\{P_n(\lambda,t)\}_{n\in \mathbb{N}}$ defined in \eqref{orthogonality condition} satisfy the time evolution \eqref{p-times}.

\end{coro}

Now we are ready to derive the noncommutative nonisospectral mixed rToda lattice with the help of the compatibility condition between three-term recurrence relation \eqref{3-term} and the time evolution \eqref{p-times}. We obtain the following theorem.

\begin{thm}\label{nnmrt}
Under assumptions \eqref{derivative-of-x} and \eqref{moments evolution}, the recurrence coefficients $\{a_n\}_{n\in \mathbb{N}}$ and $\{b_n\}_{n\in \mathbb{N}}$  defined in \eqref{an bn}
satisfy the following noncommutative nonisospectral mixed rToda  lattice
\begin{subequations}\label{rtoda}
\begin{align}
   & \frac{d}{dt}a_n=\alpha_0 a_n+\alpha_1(-a_{n+1}a_n+a_na_{n-1}+b_na_n-a_nb_{n-1})+\alpha_2(b_n^{-1}a_n-a_nb_{n-1}^{-1}),\label{rtoda1}\\
   & \frac{d}{dt}b_n=\alpha_0b_n+\alpha_1(b_na_n-a_{n+1}b_n)+\alpha_2(b_{n+1}^{-1}a_{n+1}-a_nb_{n-1}^{-1}).\label{rtoda2}
   \end{align}
\end{subequations}

\end{thm}

\begin{proof}
 We can rewrite \eqref{p-times} in matrix form as  
\begin{align}\label{lax-t}
  A \frac{d}{dt}P=LP  
\end{align}
with
\begin{equation*}
    L=\left(
\begin{array}{ccccc}
 0 &  &  &  &   \\
 f_1-\alpha_2b_{1}^{-1}a_1 & \alpha_0 &  &   &  \\
 -\alpha_2a_2b_{1}^{-1}a_1 & f_2-\alpha_2b_{2}^{-1}a_2 & 2\alpha_0 &    & \\
 &-\alpha_2a_3b_{2}^{-1}a_2 & f_3-\alpha_2b_{3}^{-1}a_3 & 3\alpha_0 &     \\
  &   &  \ddots &  \ddots & \ddots  \ \\
\end{array}
\right),
\end{equation*}
where $ f_n=(n-1)\alpha_0 a_n+\alpha_1 a_n(b_{n-1}- a_{n-1})$. 

 We now proceed to derive the noncommutative nonisospectral mixed rToda lattice.
We first consider the $(n+1)$-th row of the matrix equation \eqref{lax-t}:
\begin{align}
& (a_n, 1) \begin{pmatrix} \frac{d}{dt}P_{n-1}(\lambda,t), & \frac{d}{dt}P_{n}(\lambda,t) \end{pmatrix}^T \nonumber \\
&= (-\alpha_2 a_n b_{n-1}^{-1} a_{n-1}, \; f_n - \alpha_2 b_n^{-1} a_n, \; n\alpha_0) \begin{pmatrix} P_{n-2}(\lambda,t), & P_{n-1}(\lambda,t), & P_n(\lambda,t) \end{pmatrix}^T. \label{Thm3.3-proof-formula1}
\end{align}
Using \eqref{time evo 1}, it is straightforward to get
\begin{align}
\frac{d}{dt}P_{n-1}(\lambda,t) = \bigl((n-1)\alpha_0 - \alpha_1 a_{n-1}\bigr) P_{n-1}(\lambda,t) + \bigl(\alpha_1 a_{n-1} \lambda - \alpha_2 b_{n-1}^{-1} a_{n-1}\bigr) P_{n-2}(\lambda,t). \nonumber
\end{align}
To obtain an equation involving only $P_{n-1}(\lambda,t)$ and $P_n(\lambda,t)$, we make use of the recurrence relation \eqref{3-term} to get
\begin{equation*}
P_{n-2}(\lambda,t) = \lambda^{-1} a_{n-1}^{-1} \bigl((\lambda - b_{n-1}) P_{n-1}(\lambda,t) - P_n(\lambda,t)\bigr).
\end{equation*}
Then, combining the above two formulas, we have
\begin{align}\label{time pn-1}
\frac{d}{dt}P_{n-1}(\lambda,t)
&= \bigl((n-1)\alpha_0 - \alpha_1 a_{n-1}\bigr) P_{n-1}(\lambda,t) \nonumber \\
&\quad - \bigl(\alpha_1 a_{n-1}\lambda - \alpha_2 b_{n-1}^{-1} a_{n-1}\bigr) \lambda^{-1} a_{n-1}^{-1} \bigl((\lambda - b_{n-1}) P_{n-1}(\lambda,t) - P_n(\lambda,t)\bigr) \nonumber \\
&= \bigl(\alpha_1 - \alpha_2 \lambda^{-1} b_{n-1}^{-1}\bigr) P_n(\lambda,t) \nonumber \\
&\quad + \bigl((n-1)\alpha_0 - \alpha_1 (\lambda + a_{n-1} - b_{n-1}) - \alpha_2 (\lambda^{-1} - b_{n-1}^{-1})\bigr) P_{n-1}(\lambda,t).
\end{align}
Using the above formula, we can rewrite \eqref{Thm3.3-proof-formula1} as a $2 \times 2$ matrix identity for the time evolution of the noncommutative Laurent bi-OPs $\{P_n(\lambda,t)\}_{n \in \mathbb{N}}$:
\begin{align}\label{lax-t2}
\frac{d}{dt}\psi_n = V_n \psi_n,
\end{align}
with $\psi_n = (P_{n-1}(\lambda), P_n(\lambda))^T$ and
\begin{align}\label{Vn}
V_n = \begin{pmatrix}
(n-1)\alpha_0 - \alpha_1(\lambda - b_{n-1} + a_{n-1}) - \alpha_2(\lambda^{-1} - b_{n-1}^{-1}) & \alpha_1 - \alpha_2 \lambda^{-1} b_{n-1}^{-1} \\[4pt]
\alpha_1 \lambda a_n - \alpha_2 b_n^{-1} a_n & n\alpha_0 - \alpha_1 a_n
\end{pmatrix}.
\end{align}

Next, from the $(n+1)$-th row of the matrix equation \eqref{lax-r}, we obtain
\[
\lambda (a_n, 1) \begin{pmatrix} P_{n-1}(\lambda,t), & P_n(\lambda,t) \end{pmatrix}^T = (b_n, 1) \begin{pmatrix} P_n(\lambda,t), & P_{n+1}(\lambda,t) \end{pmatrix}^T,
\]
which can also be rewritten in the form of a $2 \times 2$ matrix as follows:
\begin{align}\label{lax-r2}
\psi_{n+1} = U_n \psi_n,
\end{align}
with
\begin{align}\label{un}
U_n = \begin{pmatrix}
0 & 1 \\
\lambda a_n & \lambda - b_n
\end{pmatrix}.
\end{align}

Finally, the compatibility condition of \eqref{lax-t2} and \eqref{lax-r2} yields
\begin{align}\label{AB}
\frac{d}{dt} U_n = V_{n+1} U_n - U_n V_n.
\end{align}
By examining the individual entries of the above equation, one easily finds that the $(1,1)$ and $(1,2)$ entries are identities, while the $(2,1)$ and $(2,2)$ entries are exactly the noncommutative nonisospectral mixed rToda lattice \eqref{rtoda1} and \eqref{rtoda2}, respectively. 

This completes the proof of the theorem.
\end{proof}

\begin{rmk}
Notably, \eqref{rtoda} represents a noncommutative nonisospectral generalized rToda  lattice, which incorporates both the positive and negative flows of the noncommutative isospectral rToda  lattice. Specifically:
\begin{itemize}
    \item  When $\alpha_1=\alpha_2=0$, \eqref{rtoda} reduces to the first flow of the noncommutative nonisospectral rToda;
    \item When $\alpha_0=\alpha_2=0$, \eqref{rtoda} reduces to the first flow of the positive flow in the noncommutative isospectral rToda lattice \cite{wang2025noncommutative};
    \item When $\alpha_0=\alpha_1=0$, \eqref{rtoda} degenerates to the first flow of the negative flow in the noncommutative isospectral rToda lattice \cite{wang2025noncommutative}.
\end{itemize}
\end{rmk}

\subsection{The matrix dP equation}\label{subsection-toda-2}
In this section, we directly perform stationary reduction on the noncommutative nonisospectral mixed rToda lattice \eqref{rtoda} and its Lax pair, obtaining the matrix dP equation and its Lax pair, respectively. In fact, we can derive the following theorem.



\begin{thm} \label{mdp}
Given assumptions \eqref{derivative-of-x} and \eqref{moments evolution}, the recurrence coefficients $\{a_n\}_{n\in \mathbb{N}}$ and $\{b_n\}_{n\in \mathbb{N}}$  defined in \eqref{an bn} satisfy the following matrix dP-type equation
\begin{subequations}\label{m-dp}
\begin{align}
   & \alpha_0 a_n+\alpha_1(-a_{n+1}a_n+a_na_{n-1}+b_na_n-a_nb_{n-1})+\alpha_2(b_n^{-1}a_n-a_nb_{n-1}^{-1})=0,\\
   & \alpha_0b_n+\alpha_1(b_na_n-a_{n+1}b_n)+\alpha_2(b_{n+1}^{-1}a_{n+1}-a_nb_{n-1}^{-1})=0.
   \end{align}
\end{subequations}
\end{thm}

\begin{proof}

Recall the Lax pair \eqref{lax-t2} and \eqref{lax-r2} for the noncommutative nonisospectral mixed rToda lattice:
\begin{align*}
\psi_{n+1} = U_n \psi_n, \qquad \frac{d}{dt}\psi_n = V_n \psi_n,
\end{align*}
where $V_n$ and $U_n$ are defined in \eqref{Vn} and \eqref{un}, respectively. Now, we consider the matrix dP-type equation based on the extended linear system
\begin{align}\label{lax-p}
\psi_{n+1} = U_n \psi_n, \qquad \frac{\partial}{\partial \lambda}\psi_n = \hat{V}_n \psi_n,
\end{align}
where $\hat{V}_n = V_n \big/ \bigl(\frac{d}{dt}\lambda\bigr) = \frac{1}{\alpha_0 \lambda} V_n$.

Taking the partial derivative with respect to $\lambda$ on both sides of the first equation in \eqref{lax-p} yields
\begin{equation*}
\frac{\partial}{\partial \lambda}\psi_{n+1} = \frac{\partial}{\partial \lambda}U_n \psi_n + U_n \frac{\partial}{\partial \lambda}\psi_n.
\end{equation*}
Using \eqref{lax-p} to replace $\psi_{n+1}$ and $\frac{\partial}{\partial \lambda}\psi_n$ by $U_n \psi_n$ and $\hat{V}_n \psi_n$, respectively, we obtain from the above equation
\begin{equation}\label{matrix dP}
\frac{\partial}{\partial \lambda}U_n = \hat{V}_{n+1} U_n - U_n \hat{V}_n.
\end{equation}
Recalling the explicit expression of $U_n$ in \eqref{un}, we have
\[
\frac{\partial}{\partial \lambda}U_n = \begin{pmatrix}
0 & 0 \\
a_n & 1
\end{pmatrix}.
\]
Together with the relation $\hat{V}_n = \frac{1}{\alpha_0 \lambda} V_n$, we obtain \eqref{m-dp} from \eqref{matrix dP}. 

This completes the proof of the theorem.
\end{proof}

\begin{rmk}
    In the commutative case, the dP-type equation \eqref{m-dp} reduces to the following form \cite{yueXL2023thesis}
\begin{subequations}\label{hxdlt}
\begin{align}
&\alpha_0 a_n+\alpha_1 a_n(a_{n-1}-a_{n+1}+b_{n}-b_{n-1})+\alpha_2 a_n\left(\frac{1}{b_{n}}-\frac{1}{b_{n-1}}\right)=0,\\
&\alpha_0 b_{n}+\alpha_1 b_{n}(a_n-a_{n+1})+\alpha_2 \left(\frac{a_{n+1}}{b_{n+1}}-\frac{a_n}{b_{n-1}}\right)=0.
\end{align}
\end{subequations}
By further simplifying \eqref{hxdlt} , the alternate dP II (alt d-P$_{\text{II}}$) \cite{nijhoff1996study}
\begin{align}
\frac{\sqrt{-\frac{\alpha_1}{\alpha_2}}\frac{n\alpha_0+\beta_0}{\alpha_1}}{1+\frac{\sqrt{-\frac{\alpha_1}{\alpha_2}}}{b_{n+1}}\frac{\sqrt{-\frac{\alpha_1}{\alpha_2}}}{b_{n}}}+\frac{\sqrt{-\frac{\alpha_1}{\alpha_2}}\frac{(n-1)\alpha_0+\beta_0}{\alpha_1}}{1+\frac{\sqrt{-\frac{\alpha_1}{\alpha_2}}}{b_{n}}\frac{\sqrt{-\frac{\alpha_1}{\alpha_2}}}{b_{n-1}}}=-\frac{\sqrt{-\frac{\alpha_2}{\alpha_1}}}{b_{n}}+\sqrt{-\frac{\alpha_1}{\alpha_2}}b_{n}+\sqrt{-\frac{\alpha_1}{\alpha_2}}\frac{n\alpha_0+\gamma_0}{\alpha_1},\label{hxdlt6}
\end{align}
can be derived, where $\beta_0=\alpha_1a_1-\alpha_2 \frac{a_1}{b_{0}b_{1}},\ \gamma_0=\alpha_1(a_1-b_0)-\alpha_2\frac{1}{b_0}$. Therefore, we can regard \eqref{m-dp} as the matrix alt d-P$_{\text{II}}$.
\end{rmk}

\begin{rmk}
   Furthermore, it is readily observable that the matrix dP-type equation \eqref{m-dp} results from the stationary reduction of the noncommutative nonisospectral mixed rToda lattice \eqref{rtoda}. This implies that \eqref{m-dp} and its Lax pair can be obtained from the stationary reduction of the noncommutative nonisospectral mixed rToda lattice \eqref{rtoda} and its corresponding Lax pair, respectively.
\end{rmk}

\begin{rmk}
    We derive the above theorem from Theorem \ref{nnmrt}, which applies to general noncommutative Laurent bi-OPs. Consequently, in addition to the matrix dP-type equations, the equations in \eqref{m-dp} also hold in a more general noncommutative setting.
\end{rmk}

\subsection{Quasideterminant solutions for the matrix dP equation}\label{subsection-toda-3}
In the previous sections, we established the  noncommutative nonisospectral mixed rToda  lattice \eqref{rtoda} by applying appropriate time evolution on the spectral parameter and moments without defining weight functions. Subsequently, a formal stationary reduction was implemented on the Lax pair of the obtained integrable equation \eqref{rtoda}, leading to the derivation of the matrix dP-type equation \eqref{m-dp}. In the next step, 
we will construct a specific weight function to produce quasideterminant solutions for matrix dP equation,
thereby justifying the validity of the stationary reduction from the solution perspective.


Since the time evolution condition \eqref{derivative-of-x} imposed on the spectral parameter, it follows readily that
\begin{equation*}
\lambda=\lambda_0e^{\alpha_0 t},
\end{equation*}
where $\lambda_0$ is the spectral parameter at the initial time and is independent of the time variable $t$. Consider the integration interval $(0,+\infty)$, then the moments can be rewritten as
\begin{align}
    m_j(t)=&\int_0^{+\infty} \lambda^jw(\lambda,t)d\lambda\nonumber\\
    =&\int_0^{+\infty} \lambda_0^je^{j\alpha_0 t}f(\lambda_0,t)d\lambda_0,\nonumber
\end{align}
with $f(\lambda_0,t)$ as the undetermined deformed weight function. Taking the derivative of the above moments with respect to time $t$ yields
\begin{equation*}
\begin{aligned}
\frac{d}{dt} m_j=j\alpha_0m_j+\int_{0} ^{+\infty} \lambda_0^j e^{j\alpha_0 t} \left(\frac{d}{dt} f(\lambda_0,t)\right)d\lambda_0.
\end{aligned}
\end{equation*}
To ensure consistency with the time evolution condition \eqref{moments evolution} set earlier, we have
\begin{equation}
    \frac{d}{dt}f(\lambda_0,t)=(\alpha_1\lambda_0 e^{\alpha_0 t}+\alpha_2\lambda_0^{-1} e^{-\alpha_0 t})f(\lambda_0,t).\nonumber
\end{equation}
From the above equation, it is easy to obtain that
\begin{equation}
    f(\lambda_0,t)=e^{\frac{\alpha_1}{\alpha_0}\lambda_0 e^{\alpha_0 t}-\frac{\alpha_2}{\alpha_0}\lambda_0^{-1} e^{-\alpha_0 t}}U,
\end{equation}
where $U$ is a matrix independent of $t$. Here, we take $U=Ve^{C\ln \lambda_0}$, allowing for appropriate choices of the constant matrices $V$ and $C$ to ensure that the moments $m_j$ remain nontrivial, which implies that it is noncommutative and cannot be diagonalized.

\begin{rmk}
An appropriate matrix valued weight function (taking a $2\times 2$ example) can be easily provided. Setting
    $$V=\left(
\begin{array}{cc}
1&0\\
1&1
\end{array}\right),\quad C=\left(
\begin{array}{cc}
1&1\\
0&1
\end{array}\right),$$
we have 
$$U=\left(
\begin{array}{cc}
\lambda_0&\lambda_0 \ln{\lambda_0}\\
\lambda_0&\lambda_0+\lambda_0 \ln{\lambda_0}
\end{array}\right),$$
then
\begin{align}
     m_j(t)=&\int_0^{+\infty} \lambda_0^je^{j\alpha_0 t}e^{\frac{\alpha_1}{\alpha_0}\lambda_0e^{\alpha_0t}-\frac{\alpha_2}{\alpha_0}\lambda_0^{-1}e^{-\alpha_0 t}}U d\lambda_0\nonumber\\
     =&\int_0^{+\infty} \lambda^je^{\frac{\alpha_1}{\alpha_0}\lambda-\frac{\alpha_2}{\alpha_0}\lambda^{-1}-\alpha_0 t}U d\lambda\nonumber.
\end{align}
\end{rmk}

\begin{lem}
  Under the assumption of \eqref{derivative-of-x} 
  together with $\frac{\alpha_1}{\alpha_0}<0$ and $\frac{\alpha_2}{\alpha_0}>0$,  the moments 
\begin{equation}\label{new moment}
     m_j(t)=\int_0^{+\infty} \lambda_0^je^{j\alpha_0 t}e^{\frac{\alpha_1}{\alpha_0}\lambda_0e^{\alpha_0t}-\frac{\alpha_2}{\alpha_0}\lambda_0^{-1}e^{-\alpha_0 t}}Ve^{C \ln{\lambda_0}}d\lambda_0
\end{equation}
satisfy the time evolution \eqref{moments evolution} and 
 \begin{equation}\label{new time}
   \frac{d}{dt}m_j(t)=-\alpha_0m_j(t)(\mathbb{I}_p+C) . 
 \end{equation}
\end{lem}

\begin{proof}
By conducting integration by parts on the moments \eqref{new moment} with respect to $\lambda_0$, we directly arrive at

\begin{align}\label{part}
m_j(t)=&\int_{0} ^{+\infty}  \lambda_0^j e^{j\alpha_0 t} e^{\frac{\alpha_1}{\alpha_0} \lambda_0 e^{\alpha_0 t}-\frac{\alpha_2}{\alpha_0}\lambda_0^{-1}e^{-\alpha_0 t}}V e^{C\ln{\lambda_0}}d\lambda_0\nonumber\\
=&\lambda_0^{j+1}e^{j\alpha_0 t} e^{\frac{\alpha_1}{\alpha_0} \lambda_0 e^{\alpha_0 t}-\frac{\alpha_2}{\alpha_0}\lambda_0^{-1}e^{-\alpha_0 t}}V e^{C\ln{\lambda_0}}|^{+\infty}_{0}\nonumber\\
&-\int_{0} ^{+\infty}j \lambda_0^je^{j\alpha_0 t} e^{\frac{\alpha_1}{\alpha_0} \lambda_0 e^{\alpha_0 t}-\frac{\alpha_2}{\alpha_0}\lambda_0^{-1}e^{-\alpha_0 t}}V e^{C\ln{\lambda_0}}d\lambda_0\nonumber\\
&-\int_{0} ^{+\infty} \lambda_0^{j+1}e^{j\alpha_0 t} (\frac{\alpha_1}{\alpha_0} e^{\alpha_0 t}+\frac{\alpha_2}{\alpha_0}\lambda_0^{-2}e^{-\alpha_0 t})e^{\frac{\alpha_1}{\alpha_0} \lambda_0 e^{\alpha_0 t}-\frac{\alpha_2}{\alpha_0}\lambda_0^{-1}e^{-\alpha_0 t}}V e^{C\ln{\lambda_0}}d\lambda_0\nonumber\\
&-\int_{0} ^{+\infty} \lambda_0^{j}e^{j\alpha_0 t} e^{\frac{\alpha_1}{\alpha_0} \lambda_0 e^{\alpha_0 t}-\frac{\alpha_2}{\alpha_0}\lambda_0^{-1}e^{\alpha_0 t}}V e^{C\ln{\lambda_0}}Cd\lambda_0\nonumber\\
=&-jm_j(t)-\frac{\alpha_1}{\alpha_0}m_{j+1}(t)-\frac{\alpha_2}{\alpha_0}m_{j-1}(t)-m_j(t)C,
\end{align}
where we used the fact at the boundary
\begin{align}
    \lim_{\lambda_0\rightarrow 0 }\lambda_0^{j+1}f(\lambda_0,t)=\lim_{\lambda_0\rightarrow +\infty }\lambda_0^{j+1}f(\lambda_0,t)=0.\nonumber
\end{align}
From equation\eqref{part}, it is clear that 
\begin{equation}\label{moment re}
    \alpha_1 m_{j+1}(t)=-\alpha_0(j+1) m_j(t)-\alpha_2m_{j-1}(t)-\alpha_0m_j(t)C.
\end{equation}
By combining equations \eqref{moment re} and \eqref{moments evolution}, one can easily obtain 

$$\frac{d}{dt}m_j(t)=-\alpha_0m_j(t)(\mathbb{I}_p+C) .$$
Therefore, we have completed the proof.
\end{proof}

Ultimately, we will demonstrate that $a_n(t)$ and $b_n(t)$ in \eqref{an bn} correspond to the solutions of the matrix dP-type equation \eqref{m-dp}. This implies that 
the noncommutative nonisospectral mixed rToda lattice and its Lax pair can indeed be stationary. This conclusion will be reached by utilizing the time evolution \eqref{new time} of the moments and the properties of the quasideterminants. This leads us to the following theorem.

\begin{thm}
  With the definition of moments \eqref{new moment}, the recurrence coefficients $\{a_n(t)\}$ and $\{b_n(t)\}$ given in \eqref{an bn} satisfy
  \begin{align}
      \frac{d}{dt}a_n(t)=\frac{d}{dt}b_n(t)=0,
  \end{align}
and the
coefficients $\gamma_{n,j}$ in the expansion $P_n(\lambda,t)=\sum_{j=0}^{n}\gamma_{n,j}(t)\lambda^j$ also satisfy
\begin{align}
      \frac{d}{dt}\gamma_{n,j}(t)=0.
  \end{align}
\end{thm}

\begin{proof}
    The central argument relies on the use of \eqref{new time}, which is valid according to the definition of the moments in \eqref{new moment}.
It is evident that from \eqref{new time} and the derivative formula for quasideterminants \eqref{derivative-formulas-for-quasideterminants}, we can obtain
\begin{equation*}
\begin{aligned}
\frac{d}{dt}H_n=\left|
\begin{array}{cc}
\Lambda_{n-1}&((\tilde{\theta}_{-n}^{n-1})^T)^{'}\\
\theta_n^n&\boxed{(m_0)^{'}}
\end{array}\right|+\sum_{k=0}^{n-1}\left(\left|
\begin{array}{cc}
\Lambda_{n-1}&((\Lambda_{n-1})_k)'\\
\theta_n^n&\boxed{((\theta_n^n)_k)'}
\end{array}\right|
\left|\begin{array}{cc}
\Lambda_{n-1}&(\tilde{\theta}_{-n}^{n-1})^T\\
e_k&\boxed{0}
\end{array}\right|\right)=-\alpha_0 H_n(\mathbb{I}_p+C),\nonumber\\
\frac{d}{dt}\tau_n=\left|
\begin{array}{cc}
\Lambda_{n-1}&((\tilde{\theta}_{1}^{n-1})^T)^{'}\\
\theta_n^n&\boxed{(m_{n+1})^{'}}
\end{array}\right|+\sum_{k=0}^{n-1}\left(\left|
\begin{array}{cc}
\Lambda_{n-1}&((\Lambda_{n-1})_k)'\\
\theta_n^n&\boxed{((\theta_n^n)_k)'}
\end{array}\right|
\left|\begin{array}{cc}
\Lambda_{n-1}&(\tilde{\theta}_{1}^{n-1})^T\\
e_k&\boxed{0}
\end{array}\right|\right)=-\alpha_0 \tau_n(\mathbb{I}_p+C).
\end{aligned}
\end{equation*}    
This results in the conclusion that
\begin{align}
  &\frac{d}{dt} a_n(t)=\left(-\frac{d}{dt}\tau_n+\tau_n\tau_{n-1}^{-1}\frac{d}{dt}\tau_{n-1}\right)\tau_{n-1}^{-1}=0,\nonumber\\
  &\frac{d}{dt} b_n(t)=a_n\left( \frac{d}{dt}H_{n-1}-H_{n-1}H_{n}^{-1} \frac{d}{dt}H_n\right)H_n^{-1}=0.\nonumber
\end{align}
From the expression for $\{P_n(\lambda)\}_{n\in N}$, it follows that
\begin{equation}
    \gamma_{n,j}=\left|
\begin{array}{cc}
\Lambda_{n-1}&(\tilde{\theta}_{-n}^{n-1})^T\\
e_j&\boxed{0}
\end{array}\right|.
\end{equation}
Similarly, we have
\begin{equation}
    \frac{d}{dt}\gamma_{n,j}=\left|
\begin{array}{cc}
\Lambda_{n-1}&((\tilde{\theta}_{-n}^{n-1})^T)^{'}\\
e_j&\boxed{0}
\end{array}\right|+\sum_{k=0}^{n-1}\left(\left|
\begin{array}{cc}
\Lambda_{n-1}&((\Lambda_{n-1})_k)'\\
e_j&\boxed{((e_j)_k)'}
\end{array}\right|
\left|\begin{array}{cc}
\Lambda_{n-1}&(\tilde{\theta}_{-n}^{n-1})^T\\
e_k&\boxed{0}
\end{array}\right|\right)=0.
\end{equation}
Therefore, we have completed the proof.
\end{proof}
From the above theorem, we can clearly observe that under the moments defined in equation \eqref{new moment}, the noncommutative nonisospectral mixed rToda
lattice \eqref{rtoda} can indeed achieve stationary reduction, and $a_n(t)$ and $b_n(t)$ in \eqref{an bn} indeed constitute solutions to the matrix dP-type equation \eqref{m-dp}. More importantly, we rigorously proved that the expansion coefficients $\gamma_{n,j}$ of the matrix Laurent bi-OPs are independent of time $t$, leading to $\frac{\partial}{\partial t}P_n=0$. This indicates that the Lax pair of the noncommutative nonisospectral mixed rToda
lattice \eqref{rtoda} is indeed stationary under definition \eqref{new moment}, from which we can derive the matrix discrete Painlev\'{e}-type equation \eqref{m-dp} and its Lax pair.

\section{Conclusion and Discussion}\label{conclusion-section}

This paper investigates noncommutative Laurent bi-OPs and matrix valued dP-type equation. Specifically, we begin with noncommutative Laurent bi-OPs and perform nonisospectral deformations without defining a specific weight function. Using the compatibility conditions of the three-term recurrence relation \eqref{3-term} and time evolution \eqref{p-times} satisfied by the noncommutative Laurent bi-OPs, the noncommutative nonisospectral mixed rToda
lattice \eqref{rtoda} is derived. Next, a formal stationary reduction of the Lax pair associated with the noncommutative nonisospectral mixed rToda lattice \eqref{rtoda} is conducted, resulting in the matrix dP-type equation \eqref{m-dp}. 
Finally, the rationality of the stationary reduction is demonstrated from the perspective of solutions by constructing a specific weight function and utilizing the properties of quasideterminants.

\section{Acknowledgement}
We are grateful to the anonymous referees for
constructive suggestions which significantly improved the manuscript. This work was supported by grants from the Research Grants Council of the Hong Kong Special Administrative Region, China (Project No. CityU 11311622, CityU 11306723 and CityU 11301924). 

\paragraph*{Data Availability}
No datasets were generated or analyzed during the current study.

\section*{Declarations}
\paragraph*{Conflict of interest}
On behalf of all authors, the corresponding author states that there is no conflict of interest.

\appendix 

\end{document}